\documentclass[a4paper,aps,showpacs,twocolumn,superscriptaddress]{revtex4-1} %APS document class style  (PRA,PRL,ETC)
  
\usepackage[utf8]{inputenc}  %Input what you want e.g., é, ł, a, ü
\usepackage[T1]{fontenc}     %Output what you want e.g., é, ł, a, ü
\usepackage[british]{babel}  %Do hyphenation according to british english
\usepackage[usenames,dvipsnames]{color} %More color choices
	\usepackage[colorlinks,citecolor=blue,linkcolor=magenta,urlcolor=blue]{hyperref}    %Hyperlinks 
\usepackage[babel]{microtype}  %Improves text justification

\usepackage{amsmath,amssymb,amsthm,bm,mathtools,amsfonts,mathrsfs,bbm,dsfont}
\usepackage{physics}

\usepackage{tikz}
\usepackage{float}
\usepackage[center]{caption}
\usepackage{subcaption}
\usetikzlibrary{positioning}
\usetikzlibrary{shapes}
\newtheoremstyle{mystyle}% name
  {3pt}%Space above
  {3pt}%Space below
  {\normalfont}%Body font
  {0pt}%Indent amount
  {\bf}% Theorem head font
  {.}%Punctuation after theorem head
  {5pt}%Space after theorem head 2
  {}%Theorem head spec (can be left empty, meaning ‘normal’)

\theoremstyle{mystyle}

\newtheorem{observation}{Observation}

\newtheorem{claim}{Claim}

\begin{document}

%spacing
\nonfrenchspacing
\allowdisplaybreaks

\title{Network Quantum Steering}

\author{Benjamin D.M. Jones}
\affiliation{H. H. Wills Physics Laboratory, University of Bristol, Bristol, BS8 1TL, UK.}
\affiliation{Quantum Engineering Centre for Doctoral Training,  University of Bristol, UK.}
\affiliation{Department of Applied Physics, University of Geneva, 1211 Geneva, Switzerland.}

\author{Ivan \v Supi\'c}
\affiliation{Department of Applied Physics, University of Geneva, 1211 Geneva, Switzerland.}
\affiliation{CNRS, LIP6, Sorbonne Universit\'{e}, 4 Place Jussieu, 75005 Paris, France.}

\author{Roope Uola}
\affiliation{Department of Applied Physics, University of Geneva, 1211 Geneva, Switzerland.}

\author{Nicolas Brunner}
\affiliation{Department of Applied Physics, University of Geneva, 1211 Geneva, Switzerland.}

\author{Paul Skrzypczyk}
\affiliation{H. H. Wills Physics Laboratory, University of Bristol, Bristol, BS8 1TL, UK.}

\date{\today}

\begin{abstract}

The development of large-scale quantum networks promises to bring a multitude of technological applications as well as shed light on foundational topics, such as quantum nonlocality. It is particularly interesting to consider scenarios where sources within the network are statistically independent, which leads to so-called network nonlocality, even when parties perform fixed measurements. Here we promote certain parties to be trusted and introduce the notion of \textit{network steering} and \textit{network local hidden state (NLHS) models} within this paradigm of independent sources. In one direction, we show how results from Bell nonlocality and quantum steering can be used to demonstrate network steering. We further show that it is a genuinely novel effect, by exhibiting unsteerable states that nevertheless demonstrate network steering, based upon entanglement swapping, yielding a form of activation. On the other hand, we provide no-go results for network steering in a large class of scenarios, by explicitly constructing NLHS models. 

\end{abstract}

\maketitle

%\textit{Introduction.---} 
The quest to deepen our understanding of quantum theory and its seemingly counter-intuitive properties has lead to many fruitful avenues of research. In particular, the phenomenon of quantum correlations have enjoyed significant attention and developments, see e.g. \cite{brunner2014bell,cavalcanti2016quantum,uola2020quantum}.

%Phenomena such as entanglement \cite{horodecki2009quantum, guhne2009entanglement}, Bell nonlocality \cite{brunner2014bell}, quantum steering \cite{uola2020quantum, cavalcanti2016quantum} and measurement incompatibility \cite{heinosaari2016invitation} have enjoyed significant attention and developments, and there now exist a wealth of tools and techniques for characterising and quantifying these topics \cite{cavalcanti2016quantitative}.

Quantum correlations expose a rich structure when considered in scenarios with many parties. A case of particular interest is that of quantum networks, featuring a number of distant parties connected by several quantum sources. Significant further work is still required to reach a deeper theoretical understanding of these scenarios, whilst also keeping inline with experimental and technological developments towards quantum networks 
%(with applications such as secure quantum communication)
\cite{wehner2018quantum}. 

Recently, a generalisation of the concept of Bell locality \cite{Bell} was proposed to tackle the question of quantum nonlocality in networks; see \cite{tavakoli2021bell} for a recent review. The key idea is to consider the various sources in the network to be statistically independent \cite{branciard2010,branciard2012bilocal,fritz2012beyond}. This independence leads to non-convexity in the space of relevant correlations, undermining the use of pre-existing tools and creating a need for new approaches, both analytically \cite{Chaves2012,tavakoli2014nonlocal,chaves2016,Rosset,Weilenmann2018,Wolfe2019,gisin2020constraints,Aberg2020,Wolfe2021} and numerically \cite{krivachy2020neural}. The network structure offers new interesting effects, such as the possibility to certify quantum nonlocality ``without inputs'' (i.e. a scenario where each party performs a fixed quantum measurement) \cite{fritz2012beyond,branciard2012bilocal,fraser2018causal,renou2019genuine,renou2020}. Also, the use of non-classical measurements allows for novel forms of quantum nonlocal correlations that are genuine to networks \cite{Supic}. In parallel, several works have  explored the structure of quantum states assuming a certain underlying network structure \cite{kraft2020quantum,Navascues,Luo2020,Spee}.

%This independence leads to non-convexity in the space of possible correlations, undermining the use of pre-existing tools and creating a need for new approaches \cite{krivachy2020neural}. A central scenario of study has been the so-called ``triangle network'' \cite{kraft2020quantum, gisin2020constraints, fraser2018causal}, consisting of three parties, pairwise connected by independent sources, each performing a fixed measurement, yielding statistics $p(a,b,c)$. Whilst it is possible to embed standard Bell nonlocality into this scenario \cite{fritz2012beyond}, the existence of correlations with a high degree of symmetry and no classical description indicate that these phenomena may be unique to networks \cite{renou2019genuine}. %It was also recently shown that network nonlocality may still be exhibited if the sources are arbitrarily close to being perfectly correlated \cite{vsupic2020quantum}.

In this work, motivated by the difficulty in characterising quantum networks both conceptually and computationally, we consider quantum network scenarios in which some of the parties are trusted while the others are untrusted. This naturally connects to the notion of quantum steering \cite{Wiseman2007} (see \cite{cavalcanti2016quantum,uola2020quantum} for reviews) which captures quantum correlations in a scenario involving a trusted and an untrusted party. While the notion of multipartite steering has been previously considered \cite{cavalcanti2015detection, he2013genuine}, our work explores a different direction, targeting the  scenario of networks with independent sources.

Our main focus here will be on the simplest setting of a linear network with trusted endpoints and intermediate untrusted parties who each perform a fixed measurement. We begin by formalising the notions of \textit{network local hidden state} (NLHS) models, and network steering. We then leverage standard steering and nonlocality scenarios to provide simple examples of network steering. Next, we outline a surprising effect in which two-way unsteerable states can demonstrate network steering through entanglement swapping, leading to a form of activation. Finally, we characterise some natural scenarios that always admit an NLHS model by identifying properties of the sources. We conclude by listing some promising future avenues for research.

\textit{Basic concepts.---} We first briefly summarise the notion of steering, as it represents the basis of what is to follow. %In a (bipartite) nonlocality scenario, two parties perform measurements on a shared state $\rho^{AB}$. The state is said to be \textit{local} if the statistics $p(a,b|x,y):=\text{Tr}(M_{a|x} \otimes M_{b|y} \rho^{AB})$ of all local measurements can be explained via a local hidden variable (LHV) model, such that  $p(a,b|x,y)=\sum_\lambda p(\lambda) ~ p(a|x,\lambda)p(b|y,\lambda)$, where $\lambda$ is a ``hidden variable'' and $p(a|x,\lambda)$ and $p(b|y,\lambda)$ are Alice's and Bob's local response functions. If on the other hand there exist measurements such that $p(a,b|x,y)$ does not admit such a decomposition, we say that the state $\rho^{AB}$ is \textit{nonlocal} \cite{brunner2014bell}.

In a (bipartite) steering scenario, one party performs measurements on a shared state $\rho^{AB}$, which `steers' the quantum state of the other particle. If Alice performs a set of measurements, labelled by $x$, with outcomes $a$, and corresponding POVM elements $M_{a|x}$, then 
 the collection of sub-normalised `steered states' of Bob are $\sigma_{a|x}^B:=\text{Tr}_A(M_{a|x}^{A} \otimes \mathbbm{1}^B \rho^{AB})$, where $p(a|x) = \text{Tr}(\sigma_{a|x})$ are the statistics of Alice's measurements. The collection of sub-normalised states $\{ \sigma_{a|x} \}_{a,x}$ are commonly referred to as an \textit{assemblage} \cite{pusey2013negativity}. If the assemblage can be explained by a \textit{local hidden state} (LHS) model, of the form $\sigma_{a|x} = \sum_\lambda p(\lambda) ~ p(a|x,\lambda)\sigma_\lambda$, where $\lambda$ is a hidden variable, distributed according to $p(\lambda)$, $\sigma_\lambda$ are `hidden states' of Bob, and $p(a|x,\lambda)$ are local `response functions' of Alice, then we say that it has LHS form, or does not demonstrate steering \cite{Wiseman2007}. If there exist measurements such that $\sigma_{a|x}$ does not admit such an LHS decomposition, we say that the state $\rho^{AB}$ is \textit{steerable} from $A$ to $B$. If for all measurements we can never demonstrate steering with a given state, we say it is unsteerable (from $A$ to $B$) \footnote{Note that steering can be asymmetrical; some states are steerable from Alice to Bob, but not the other way around \cite{Bowles}.}.

\textit{Network Steering.---} We will now introduce our main new notion, that of network steering. Here, we have a collection of independent sources which distribute quantum states to a subset of parties. In the standard network nonlocality scenario all parties are assumed to be untrusted, and to perform `black-box' measurements. Here, in contrast, inspired by the steering scenario, we will consider only a subset of the parties to be untrusted, and the remainder trusted. We will be interested in the (sub-normalised) states that are prepared for the trusted parties by the measurements of the untrusted parties. We refer to this general set-up as \emph{network steering}. 

We focus primarily on a simple scenario, with $n$ parties arranged in a line, where the endpoint parties are  trusted, and intermediate parties are untrusted and each perform a single, fixed measurement. The simplest such scenario has three parties and two sources (see Fig. \ref{fig:bilocal_scenario}), as in entanglement swapping \cite{Ekert}. Here the first two parties share a state $\rho^{AB}$ and the second and third parties share a state $\rho^{B'C}$, and the central party performs a fixed measurement $M^{BB'}_b$. The sub-normalised states prepared for $A$ and $C$ by this measurement are
\begin{equation}
\sigma_{b}^{AC} = \text{Tr}_{BB'}\Big ( \Big [\mathbbm{1}^A \otimes M^{BB'}_b \otimes \mathbbm{1}^C \Big ] \rho^{AB} \otimes \rho^{B'C} \Big ), \label{eq:bilocalquantum}
\end{equation}
which occur with probability $p(b) = \text{Tr}(\sigma_{b}^{AC}).$ We will refer to $\{\sigma_b \}_b$  as a \textit{network assemblage}.

In order to determine when this network assemblage demonstrates network steering we need to introduce the notion of a \textit{network local hidden state} (NLHS) model, which takes the form
\begin{equation}
    \sigma^{AC}_{b}=\sum_{\beta,\gamma}p(\beta)p(\gamma)~p(b|\beta,\gamma) ~ \sigma^{A}_\beta\otimes\sigma^{C}_\gamma, \label{eq:bilocallhs}
\end{equation}
where $\beta$ and $\sigma_\beta^A$ are the hidden variable and hidden states of the first source, $\gamma$ and $\sigma_\gamma^C$ those of the second source, and $p(b|\beta, \gamma)$ the local response function of Bob. If there is no such model that can explain the network assemblage $\sigma_b$, then we say it demonstrates \textit{network steering}.  Interestingly, whereas conventional quantum steering requires multiple measurements to be performed by the untrusted party, just as with network nonlocality, we shall see here that even a fixed measurement can suffice to demonstrate network steering.

\begin{figure}
  \begin{subfigure}[t]{.5\columnwidth}
  \centering
    \caption{\hspace*{32pt}}
    \vspace{-12pt}
  \scalebox{0.8}{\hspace*{-37pt}
  \includegraphics[]{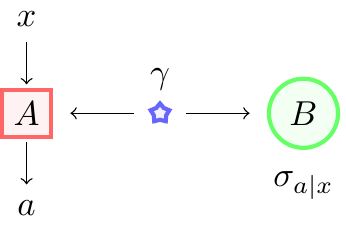}
}
  \label{fig:steer_stand}
\end{subfigure}%
\begin{subfigure}[t]{.5\columnwidth}
  \centering
    \caption{\hspace*{15pt}}
    \vspace{-2pt}
  \scalebox{0.8}{\hspace*{-15pt}
  \includegraphics[]{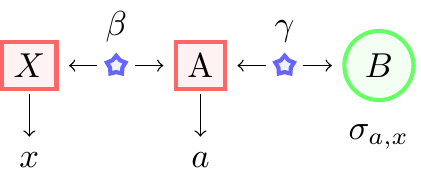}
}
  \label{fig:steer_noinputs}
\end{subfigure}

    \vspace{-6pt}
    
  \begin{subfigure}[t]{.5\columnwidth}
  \centering
    \caption{\hspace*{25pt}}
    %\vspace{-10pt}
  \scalebox{0.7}{\hspace*{-37pt}
  \includegraphics[]{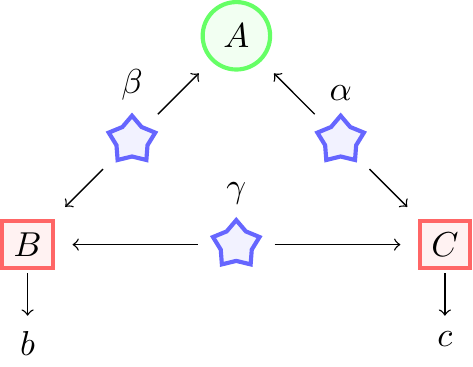}
}
  \label{fig:triangle_scenario}
\end{subfigure}%
\begin{subfigure}[t]{.5\columnwidth}
  \centering
  \vspace{10pt}
    \caption{\hspace*{15pt}}
    %\vspace{2pt}
  \scalebox{0.8}{\hspace*{-25pt}
  \includegraphics[]{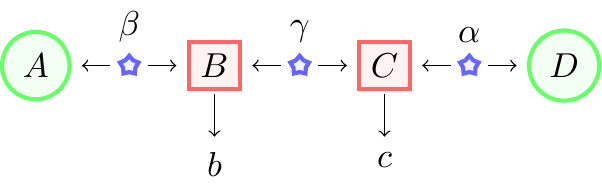}
}
  \label{fig:triangle_line_variable_scenario}
\end{subfigure}

    \vspace{-7pt}

  \begin{subfigure}[t]{.5\columnwidth}
  \centering
    \caption{\hspace*{30pt}}\vspace{-3pt}
\scalebox{0.8}{\hspace*{-37pt}
\vspace{-2pt}
\includegraphics[]{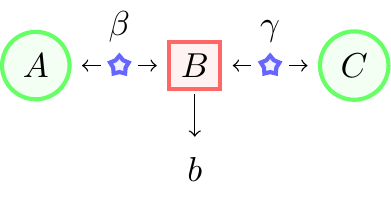}
}
  \label{fig:bilocal_scenario}
\end{subfigure}%
\begin{subfigure}[t]{.5\columnwidth}
  \centering
    \caption{\hspace*{15pt}}
    
\scalebox{0.8}{\hspace*{-18pt}
\includegraphics[]{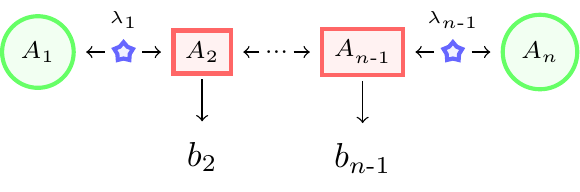}
}
  \label{fig:line_scenario}
\end{subfigure}
\vspace{-15pt}
    \caption{ Network steering scenarios. Green circles represent trusted parties, and red squares represent untrusted parties. (a) Standard steering scenario. (b) Steering scenario without inputs. (c) Triangle scenario with a trusted party. (d) Triangle scenario interpreted as a line. (e) Entanglement swapping scenario with trusted endpoints. (f) Generalised line scenario with trusted endpoints.
    }
        \label{fig:various scenarios}
\end{figure}

We note first in (\ref{eq:bilocallhs}) that each $\sigma_b^{AC}$ is in fact separable. Thus the presence of entanglement in any single $\sigma_b$ suffices to rule out an NLHS model, and therefore demonstrates network steering. %Also note that for both the NLHS and quantum model, $\sum_b \sigma_b$ is a product state.

The above generalises in a natural way to the $n$-party line network depicted in Fig. \ref{fig:line_scenario}, with outcomes $b_2,\dots , b_{n-1}$. We explicitly include the straightforward generalisation of (\ref{eq:bilocalquantum}) and (\ref{eq:bilocallhs}) in  Appendix \ref{app:further-obvs}, and see that the following observation holds generally:

\begin{observation}
For any linear network with trusted endpoints, the entanglement of a single $\sigma_{b_2, \dots, b_{n-1}}$ is sufficient to rule out an NLHS model, and thus demonstrate network steering.
\end{observation}

%\begin{observation}
%For any linear network structure with trusted endpoints, summing over any output $b_i$ results in a product %, namely $\sum_{b_i} \sigma_{b_2, \dots, b_{n-1}}$ is product for any $i$.
%\end{observation}

For more general networks, we can represent them as undirected graphs, where each node is either untrusted or trusted, and the edges represent independent sources. If all the parties are untrusted, the quantity of interest is the observed statistics $p(a,b,\dots|x,y,\dots)$. When at least one party is trusted this is replaced by some network assemblage $\sigma_{a,b,\dots|x,y,\dots}$. A key observation that will prove useful is the following equivalence between networks, a generalisation from the network nonlocality case \cite{fritz2012beyond}:

\begin{observation}
Any network with an untrusted party $A$ that has an input $x$, received with probability $p(x)$, and outcome $a$, is equivalent to a network with an additional untrusted party $A'$ who shares an additional source with $A$, neither of whom now has an input.  In this new network, the outcome of $A'$ is $x$, the old input of $A$. The relation between the network assemblages in the first and second scenarios are $p(x)\sigma_{a,\dots|x,\dots}^{A\dots} = \sigma_{a,x,\dots}^{AA'\dots} $.
\end{observation}
By virtue of the fact that quantum mechanics admits local tomography, we also note the following:

\begin{observation}
A trusted party connected to $n$ independent sources can without loss of generality be replaced by $n$ endpoint trusted parties, each connected to a single source.
\end{observation}
\noindent This allows us, for example, to interpret linear networks as rings with a single trusted party -- e.g.~the four party linear network with trusted endpoints can also be viewed as the triangle network %\cite{kraft2020quantum}
where one of the parties is trusted, as in Figs.~\ref{fig:triangle_scenario} and \ref{fig:triangle_line_variable_scenario}. This observation motivates our choice to focus our discussion on linear networks, which we understand now to be relevant for more complex, non-linear networks. We detail further basic observations in Appendix \ref{app:further-obvs}.

\textit{Demonstrating Network Steering.---} We now begin our exploration of demonstrating network steering, and explain how and when steerable states will lead to network steering when placed in a network. We consider first the scenario of Fig. \ref{fig:bilocal_scenario}. If one source distributes a state which is steerable in the standard steering scenario, 
then Observation 2 would seem to indicate that even if the second source distributes only separable states (which we will refer to as a \emph{separable source}), it should still be possible to use this to encode `the input' to the measurement, and thus demonstrate network steering. Here we make this intuition precise.

Consider network scenario depicted in Fig.~\ref{fig:steer_noinputs}, with two untrusted parties without inputs steering a third, leading to a network assemblage $\sigma_{a,x}$. Here the NLHS condition reads
\begin{equation}
	\sigma_{a,x} = \sum_{\beta, \gamma} p(\beta) p(\gamma) ~ p(x|\beta) p(a|\beta, \gamma) \sigma_\gamma. \label{eq:steernoinputslhs}
\end{equation}
We can then observe the following:
\begin{claim}
	If $\sigma_{a,x}$ has an NLHS model, then $\sigma_{a|x} := \sigma_{a,x}/p(x)$ has an LHS model, where $p(x) = \text{Tr}\sum_a \sigma_{a,x}$.
\end{claim}

\begin{proof}
	We can write (\ref{eq:steernoinputslhs}) as
	\begin{equation}
		\sigma_{a,x}  = p(x)\sum_{ \gamma}  p(\gamma) ~ p(a|x, \gamma), \sigma_\gamma
	\end{equation}
	where $p(x):=\text{Tr}(\sum_a \sigma_{a,x} ) = \sum_\beta p(\beta)p(x|\beta)$ and $p(a|x,\gamma) := \frac{1}{{p(x)}} \sum_\beta p(\beta)p(x|\beta)p(a|\beta, \gamma)$. 
	The result then follows.
\end{proof}

\noindent This is an analogous result to that proved in \cite{fritz2012beyond} relating Bell scenario statistics $p(a,b|x,y)$ to network nonlocality statistics $p(a,b,x,y)$, the corresponding distribution without inputs. We link this to the scenario from Fig.~\ref{fig:bilocal_scenario} where both endpoints are trusted.

\begin{claim}
	If $\sigma_b$ has an NLHS model, then $\sigma_{b, x} := \text{Tr}_A([M_x^A \otimes \mathbbm{1}^C ]\sigma_b )$ has an  NLHS model, for any measurement $M_x$. 
\end{claim}

\begin{proof}
	When $\sigma_{b}$ has an NLHS model of the form \eqref{eq:bilocallhs}, it follows that
	\begin{equation}\label{eq:nLHStoLHS}
		\sigma_{b, x} = \sum_{\beta, \gamma} p(\beta) p(\gamma) ~ \text{Tr}(M_x \sigma_\beta) p(b|\beta, \gamma) \sigma_\gamma,
\end{equation}
which is an NLHS model of the form \eqref{eq:steernoinputslhs}, with $p(x|\beta) := \text{Tr}(M_x\sigma_\beta)$.
$\qedhere$
\end{proof}

Putting this together, suppose that $\rho^{B'C}$ is steerable, such that $\sigma_{b|x}:=\text{Tr}([M_{b|x}\otimes \mathbbm{1}] \rho^{B'C})$ demonstrates steering for some $M_{b|x}$. Let $\rho^{AB} = \sum_x \frac{1}{d} \ketbra{x}{x}\otimes \ketbra{x}{x}$ where $d$ is the number of measurements $x$, and $\{\ket{x}\}_x$ form an orthonormal basis, and $M_b = \sum_{x'} \ketbra{x'}{x'} \otimes M_{b|x'}$. The resulting network assemblage $\sigma_b$, from \eqref{eq:bilocalquantum}, is seen to be
\begin{equation}
	\sigma_b %&= \text{Tr}_{BB'}\bigg ( \bigg [\mathbbm{1}^A \otimes M^{BB'}_b \otimes \mathbbm{1}^C \bigg ] \rho^{AB} \otimes \rho^{B'C} \bigg ) \\
	%&= \sum_x \frac{1}{d} \ketbra{x}{x} \otimes \text{Tr}(M_{b|x}\otimes \mathbbm{1} \rho^{B'C}) \\
	= \sum_x \frac{1}{d} \ketbra{x}{x} \otimes \sigma_{b|x}. 
\end{equation}
Now, from the above claims we can see that this must demonstrate network steering. Indeed, if instead it had an NLHS model, then from Claim 2, $\sigma_{b,x}:=\text{Tr}_A([\ketbra{x}{x} \otimes \mathbbm{1}^C ] \sigma_b ) = \frac{1}{d}\sigma_{b|x}$ would have an  NLHS model with $p(x) = 1/d$. Then, from Claim 1, $\sigma_{b,x}$ would have an LHS model, but by assumption it does not. This shows that all steerable states lead also to network steering when placed in a network with an appropriate separable state. Interestingly, this occurs even though $\sigma_b$ is separable.

Similar arguments apply for showing that in the line with four parties from Fig.~\ref{fig:triangle_line_variable_scenario}, we can always demonstrate network steering when the central state is nonlocal, and the adjacent endpoint sources are suitable separable states, providing the inputs. That is, if $\sigma_{b,c}$ has an NLHS model, then by $A$ and $D$ applying measurements $M_x$ and $M_y$ the associated probability distributions $p(b,c,x,y)$ and $p(b,c|x,y)$ necessarily have NLHV and LHV models respectively (see \cite{fritz2012beyond}). So for any nonlocal central source, we can find appropriate measurements and adjacent separable sources such that $\sigma_{b,c}$ demonstrates network steering.

\textit{Activation.---} The above constructions of network steering relied on steering or nonlocality in standard scenarios. Here we show that network steering is possible even when using only (two-way) unsteerable states, which can be viewed as a form of activation. Note that this complements previous examples of activation of steering in the standard bipartite scenario \cite{quintino2016superactivation}.

We define the Doubly-Erased Werner (DEW) state as the two-qubit Werner state after both subsystems have undergone an identical erasure channel:
\begin{equation}
	\rho_\text{\scalebox{.7}[1.0]{\tiny DEW}}(\eta, \omega) := \Lambda_\eta \otimes \Lambda_\eta \bigg ( \omega \ketbra{\psi^-}{\psi^-} + (1-\omega)\frac{\mathbbm{1}}{4} \bigg ),
\end{equation}
where $\ket{\psi^-} = (\ket{01} - \ket{10})/\sqrt{2}$, and $\Lambda_\eta (\rho) = \eta\rho + (1-\eta) \ketbra{2}{2}$, where $\ket{2}$ represents the loss of the system. $\rho_\text{\scalebox{.7}[1.0]{\tiny DEW}}(\eta, \omega)$ is entangled when $\omega > \frac{1}{3}$ (and $\eta \neq 0$), and is unsteerable (in both directions) when 
$\eta \leq \frac{2}{3} (1- \omega)$.
This follows from \cite{tischler2018conclusive}, as for any state $\rho^{AB}$ unsteerable from Alice to Bob, we have that $\mathbbm{1}^A\otimes \Omega^B [ \rho^{AB} ]$ is also unsteerable from $A$ to $B$, for any channel $\Omega$ \cite{quintino2015inequivalence}. Note also that in the context of entanglement swapping, projecting two copies of $\rho_\text{\scalebox{.7}[1.0]{\tiny DEW}}(\eta, \omega)$ onto $\ketbra{\psi^-}{\psi^-}$ leads to $\rho_\text{\scalebox{.7}[1.0]{\tiny DEW}}(\eta, \omega^2)$ (with probability $\eta^2/4$), that is to a DEW state with squared visibility (See Appendix \ref{app:ew-dew} for details).

Consider now the line network from Fig.~\ref{fig:line_scenario} with each source distributing a copy of  $\rho_\text{\scalebox{.7}[1.0]{\tiny DEW}}(\eta, \omega)$, and all untrusted parties performing the fixed measurement  $M_0 = \ketbra{\psi^-}{\psi^-}$, $M_1 =  \mathbbm{1}-\ketbra{\psi^-}{\psi^-}$, leading to the network assemblage $\sigma_{b_2,\ldots,b_{n-1}}$. Now, if we choose $\eta = \frac{2}{3}(1-\omega)$ and $1>\omega>(\frac{1}{3})^{\frac{1}{n}}$, then each DEW is entangled but unsteerable, and we find, due to the entanglement-swapping property noted above, that the element $\sigma_{0,\ldots,0}$ (corresponding to a successful swap in each case), will be proportional to the state $\rho_\text{\scalebox{.7}[1.0]{\tiny DEW}}(\eta, \omega')$ with $\omega' > \frac{1}{3}$, and therefore entangled. From Observation 1, this precludes an NLHS model description, and therefore demonstrates network steering, even though each DEW state was unsteerable.

\begin{figure}[t]
	\centering
	\begin{subfigure}{.5\textwidth}
		\centering
		\caption{\vspace{-24pt} \hspace*{200pt}}
		\label{fig:lhs_bilocal}
		\hspace*{-5pt}\includegraphics[]{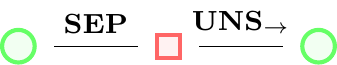}
	\end{subfigure}
	\vspace{0.2cm}
	
	\begin{subfigure}{.5\textwidth}
		\centering
		\caption{\vspace{-48pt} \hspace*{200pt}}
		\label{fig:lhs_triangle}
		\includegraphics[]{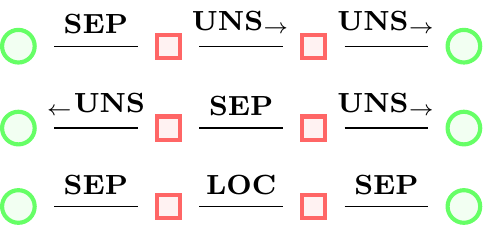}		
	\end{subfigure}
	\vspace{0.2cm}
	
	%\begin{subfigure}{.5\textwidth}
	%  \centering
	%\includegraphics[]{lhs_line}
	%  \caption{General scenarios and examples.}
	%  \label{fig:lhs_line}
	%\end{subfigure}
	
	\caption{
	Classifying the structure of some NLHS models. Green circles represent trusted parties, and red squares represent untrusted parties who perform a fixed measurement. (a) In the scenario of Fig.~\ref{fig:bilocal_scenario}, when one source is separable ($\mathbf{SEP}$), this acts as an input to the adjacent measurements, and by taking the second source as unsteerable ($\mathbf{UNS}_\rightarrow$) then this always leads to an NLHS model. (b)  Similar results hold in the ``unwrapped'' triangle scenario (Fig.~\ref{fig:triangle_line_variable_scenario}), where now sources can also be taken as local, ($\mathbf{LOC}$). We expand and detail this further in Appendix \ref{app:nlhs-on-line}.}  
	\label{fig:lhsmodels}
\end{figure}

\textit{Simple NLHS models.---} We finish our exploration by considering to what extent the properties of the quantum sources directly affect the possibility of an NLHS model. We will refer to a source as being separable, unsteerable or local if it is only capable of generating separable, unsteerable or local states respectively. As an illustrative example, in the three-party scenario of Fig.~\ref{fig:bilocal_scenario} if one source is separable and the other source is unsteerable (towards the trusted party), then for any fixed central measurement the network assemblage $\sigma_b$ will always be NLHS. Indeed taking $\rho^{AB} = \sum_\gamma p(\gamma) \sigma^A_\gamma \otimes \sigma^B_\gamma$, and inserting into (\ref{eq:bilocalquantum}) gives
\begin{equation}
\sigma_{b} = \sum_\gamma p(\gamma)~ \sigma_\gamma^A \otimes \text{Tr}_{BB'}\Big ( \Big [ M_b \otimes \mathbbm{1}^C \Big ] \sigma^B_\gamma \otimes  \rho^{B'C} \Big ).
\end{equation}
Defining  $M_{b|\gamma} := \text{Tr}_{B} (M_b [ \sigma^B_\gamma \otimes \mathbbm{1}^{B'}  ] )$ which form a set of valid measurement operators leads us to write
\begin{equation}
    \sigma_b = \sum_\gamma p(\gamma)~ \sigma_\gamma^A \otimes \text{Tr}_{B'}\Big ( \Big [ M_{b|\gamma}\otimes \mathbbm{1}^C \Big ]  \rho^{B'C}  \Big ).
\end{equation}
If $\rho^{B'C}$ is unsteerable from $B'$ to $C$, this allows us to extract a LHS model, yielding
\begin{align}
    \sigma_b &= \sum_\gamma p(\gamma)~ \sigma_\gamma^A \otimes \bigg ( \sum_\lambda p(\lambda)~p(b|\lambda, \gamma) \sigma_\lambda^C  \bigg ) \\
    &= \sum_{\gamma, \lambda} p(\gamma)p(\lambda) ~ p(b|\lambda, \gamma)\sigma_\gamma^A  \otimes \sigma_\lambda^C , 
\end{align}
which is an NLHS model (\ref{eq:bilocallhs}). Hence the combination of a separable and unsteerable source (to the trusted party) can never lead to network steering, as shown in Fig.~\ref{fig:lhs_bilocal}.

Similar results follow in more complicated scenarios. In Fig.~\ref{fig:lhsmodels} (b) we give the three configurations which always lead to NLHS models in the (unwrapped) triangle scenario of Fig.~\ref{fig:triangle_line_variable_scenario}, and we give further generalisations for the line scenario of Fig.~\ref{fig:line_scenario} in Appendix \ref{app:nlhs-on-line}. The main concept behind all of these results is that separable and unsteerable sources provide a form of input, allowing us to write down large classes of non-trivial NLHS models.

\textit{Conclusions.---} We have introduced the notions of network steering and network local hidden state models. We discussed illustrative examples, and showed that the network scenario leads to a form of activation of steering. Finally, we have started a characterisation of NLHS models based solely upon properties of the sources. There are many fascinating and novel future questions to tackle.

First, it would be interesting to determine if either NLHS assemblages or the full set of network assemblages can be characterized via techniques based on semi-definite programming, using for instance the approach of \cite{Wolfe2021}. A related direction is to further classify NLHS models based on the properties of the sources. For instance, consider four parties sharing \textit{separable}, \textit{local} and \textit{unsteerable} sources, or five parties sharing \textit{separable}, \textit{local}, \textit{local}, and \textit{separable} sources. In neither of these cases do we currently know if network steering can arise or not. 

%power of network steering. Our example involving the DEW state relied on entanglement in the network assemblage, while our examples based upon steering did not involve entanglement, thus there is still much to be understood. Moreover, our analysis here is not yet quantitative, and it will be interesting to develop inequalities to detect network steering.

%tools analogous to entanglement witnesses \cite{guhne2009entanglement}, steering inequalities \cite{cavalcanti2016quantum}, and Bell inequalities \cite{brunner2014bell}.

%Our characterisation of NLHS models based upon sources is also only partial at this stage, and it would be interesting to fully classify this. Simple examples include four parties sharing \textit{separable}, \textit{local} and \textit{unsteerable} sources, or five parties sharing \textit{separable}, \textit{local}, \textit{local}, and \textit{separable} sources. In neither of these cases do we currently know if network steering can arise or not. 

Here we have focused primarily on the properties of the sources, but it would also be interesting to consider the measurements, and understand which of their properties (e.g.~entanglement or incompatiblity) are relevant for network steering. Future work could also consider the significance of our work for quantum repeaters \cite{wehner2018quantum}, explore links with superactivation of quantum steering \cite{quintino2016superactivation}, or extend recent work on post-quantum steering \cite{sainz2015postquantum} to this setting.

Finally, our initial motivation for this work was to attempt to gain clarity on network nonlocality problems, such as those in the triangle network. It is our hope that developing our framework further will lead to discovering novel nonlocal correlations, unique to networks.

\textit{Acknowledgements.---} We thank Marco Túlio Quintino for helpful discussions. BDMJ acknowledges support from UK EPSRC (EP/SO23607/1); PS from a Royal Society URF (UHQT); IS, RU and NB from the Swiss National Science Foundation (project 2000021 192244/1 and NCCR SwissMap). \\

\bibliographystyle{apsrev4-1}
\bibliography{bib.bib}

\onecolumngrid
\appendix
\captionsetup{justification=centering}

\section{Network Steering and NLHS models on the line}
\label{app:nlhs-on-line}

\subsection{The Simplest Scenario}

In the main text we mainly discuss the scenario with three parties and trusted endpoints. Here we will extend and generalise this, first to four parties on the line, which we can also interpret as the triangle with a single trusted party. We then generalise our discussion to lines (equivalently, rings) of arbitrary length. To fix notation, Greek subscripts will denote random (hidden) variables (such as $\alpha$ in $\sigma_\alpha$) and Roman subscripts (such as the network assemblage $\sigma_b$, or the measurement $M_c$) will denote outcome index labels. Subsystem labels will be denoted by superscripts, for example $\rho^{AB}$ is a quantum state on subsystems $A$ and $B$, and $M_b^{BB'}$ is a measurement on systems $B$ and $B'$ (with outcomes $b$) -- here subsystems $B$ and $B'$ are implicitly assumed to belong to the same party.

Recalling the 3 party simple scenario, such a scenario is described through quantum mechanics as the existence of quantum sources $\rho^{AB}$, $\rho^{B'C}$ and a fixed measurement on the central party $M_b^{BB'}$ such that the resulting set of states can be written as
\begin{equation}
\sigma_{b}^{AC} = \text{Tr}_{BB'}\bigg ( \bigg [\mathbbm{1}^A \otimes M^{BB'}_b \otimes \mathbbm{1}^C \bigg ] \rho^{AB} \otimes \rho^{B'C} \bigg ). \end{equation}
As in the main text, our definition of a NLHS model here is the existence of probability distributions $p(\alpha)$, $p(\gamma)$ and $p(b|\alpha, \gamma)$, and normalised states $\sigma^A_\alpha$, $\sigma^C_\gamma$ such that

\begin{equation}
\sigma^{AC}_{b} =\sum_{\alpha,\gamma} ~ p(\alpha)p(\gamma) ~ p(b|\alpha,\gamma) ~ \sigma^{A}_\alpha\otimes\sigma^{C}_\gamma. 
\end{equation}

\subsection{The Triangle Scenario}

We can naturally extend this to the line with four parties and trusted endpoints, equivalently viewing this as the triangle network with a single trusted party (Figure \ref{fig:app_triangle_scenario}).
\begin{figure}[H]
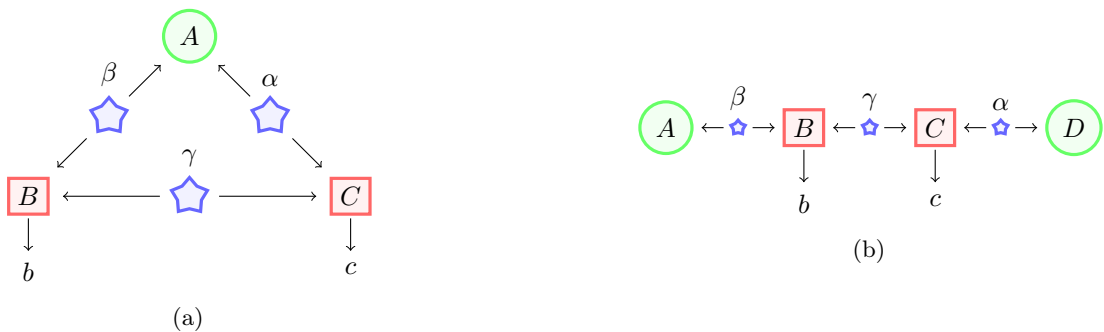

    \centering
    \begin{subfigure}{.5\linewidth}
    \centering
      \includegraphics[]{triangle_scenario}
      \caption{}
    \end{subfigure}%
    \begin{subfigure}{.5\linewidth}
    \centering
      \includegraphics[]{triangle_line_variable_scenario}
      \caption{}
    \end{subfigure}
\caption{Line scenario with four parties, or alternatively the triangle scenario with a single trusted party.}
\label{fig:app_triangle_scenario}
\end{figure}

 Here the quantum description would be
\begin{equation}
\sigma_{b,c}^{AD} = \text{Tr}_{BB'CC'}\bigg ( \bigg [\mathbbm{1}^A \otimes M^{BB'}_b \otimes M^{CC'}_c \otimes \mathbbm{1}^D \bigg ]\rho^{AB} \otimes \rho^{B'C} \otimes \rho^{C'D} \bigg ), \label{eq:app_trianglequantum}
\end{equation}

and the network assemblage $\sigma_{b,c}^{AD}$ would admit an NLHS description if it could be written in the form
\begin{equation}
    \sigma_{b,c}^{AD}=\sum_{\alpha,\beta,\gamma}p(\alpha)p(\beta)p(\gamma)p(a|\beta,\gamma)p(b|\alpha,\gamma)\sigma^A_\alpha\otimes\sigma^D_\beta. \label{eq:app_trianglelhs}
\end{equation}

We will now consider how NLHS models can naturally arise by considering properties of the three sources. If the central source is separable, i.e. $\rho^{B'C} = \sum_\gamma p(\gamma) \sigma^{B'}_\gamma \otimes \sigma^C_\gamma$. Inserting this into Equation (\ref{eq:app_trianglequantum}) yields
\begin{align}
\sigma_{b,c}^{AD} &= \text{Tr}_{BB'CC'}\bigg ( \bigg [\mathbbm{1}^A \otimes M^{BB'}_b \otimes M^{CC'}_c \otimes \mathbbm{1}^D \bigg ] \rho^{AB} \otimes \rho^{B'C} \otimes \rho^{C'D} \bigg ) \label{eq:app_triangle-first}\\
&=\sum_\gamma p(\gamma)~\text{Tr}_{BB'CC'}\bigg ( \bigg [\mathbbm{1}^A \otimes M^{BB'}_b \otimes M^{CC'}_c \otimes \mathbbm{1}^D \bigg ] \rho^{AB} \otimes \sigma^{B'}_\gamma \otimes \sigma^C_\gamma \otimes \rho^{C'D} \bigg ) \\
&= \sum_\gamma p(\gamma)~ \text{Tr}_{BB'}\bigg ( \bigg [M_b^{BB'} \otimes \mathbbm{1}^A \bigg ]  \rho^{AB} \otimes \sigma_\gamma^{B'} \bigg )\otimes \text{Tr}_{CC'}\bigg ( \bigg [ M_c^{CC'} \otimes \mathbbm{1}^D \bigg ] \sigma^C_\gamma \otimes \rho^{C'D} \bigg ) \\
&= \sum_\gamma p(\gamma)~ \text{Tr}_{B}\bigg ( \bigg [M_{b|\gamma}^B \otimes \mathbbm{1}^A \bigg ]  \rho^{AB} \bigg )\otimes \text{Tr}_{C'}\bigg ( \bigg [ M_{c|\gamma}^{C'} \otimes \mathbbm{1}^D \bigg ] \rho^{C'D} \bigg ),
\end{align}
where we defined $M_{b|\gamma}^B := \text{Tr}_{B'}\bigg (M_b^{BB'} \mathbbm{1}^B \otimes \sigma^{B'}_\gamma  \bigg )$ and $M_{c|\gamma}^{C'} := \text{Tr}_{C}\bigg (M_c^{CC'}  \sigma^{C}_\gamma \otimes \mathbbm{1}^{C'}   \bigg )$ as valid sets of measurements. Then if $\rho^{AB}$ and $\rho^{C'D}$ are \textit{unsteerable} towards $A$ and $D$ respectively (but possibly entangled: see the main text and \cite{uola2020quantum}), we can extract a \textit{local hidden state} (LHS) model to obtain
\begin{align}
    \sigma_{b,c}^{AD} &= \sum_\gamma p(\gamma)~ \bigg ( \sum_\alpha p(\alpha) p(b|\alpha,\gamma) \sigma_\alpha^A \bigg )\otimes \bigg ( \sum_\beta p(\beta) p(c|\beta,\gamma) \sigma_\beta^D \bigg ) \\
    &= \sum_{\alpha, \beta , \gamma} p(\alpha) p(\beta) p(\gamma)~ p(b|\alpha,\gamma)p(c|\beta,\gamma)~ \sigma_\alpha^A \otimes \sigma_\beta^D,
\end{align}
which has exactly the same form as the NLHS condition in Equation (\ref{eq:app_trianglelhs}). Therefore taking $\rho^{AB}$ as separable and $\rho^{AB}$ and $\rho^{C'D}$ unsteerable towards $A$ and $D$ respectively, we will always arrive at an NLHS model, for any intermediate measurements $M_b^{BB'}$ and $M_c^{CC'}$.

Similarly suppose now that the source $\rho^{AB} = \sum_\alpha p(\alpha) \sigma^A_\alpha \otimes \sigma^B_\alpha$ is separable. Then we find
\begin{align}
\sigma_{b,c} &= \sum_\alpha p(\alpha)~ \sigma_\alpha^A \otimes \text{Tr}_{BB'CC'}\bigg ( \bigg [M_b^{BB'} \otimes M_c^{CC'} \bigg ]  \sigma^{B}_\alpha  \otimes \rho^{B'C} \otimes \rho^{C'D} \bigg ) \label{eq:app_oppsep}
\end{align}%
If  $\rho^{C'D}$ is also separable, and $\rho^{B'C}$ is \textit{local} (in the Bell nonlocality sense, see the main text and \cite{brunner2014bell}), we get
\begin{align}
\sigma_{b,c} &= \sum_{\alpha,\beta} p(\alpha) p(\beta)~ \text{Tr}_{BB'CC'}\bigg (  \bigg [ M_b^{BB'} \otimes M_c^{CC'} \bigg ]  \sigma^{B}_\alpha  \otimes \rho^{B'C} \otimes \sigma^{C'}_\beta \bigg ) \sigma^A_\alpha \otimes \sigma^D_\beta \\
&= \sum_{\alpha,\beta} p(\alpha) p(\beta)~ \text{Tr}_{BB'CC'}\bigg (  \bigg [ M_{b|\alpha}^{B'} \otimes M_{c|\beta}^{C} \bigg ]   \rho^{B'C} \bigg ) \sigma^A_\alpha \otimes \sigma^D_\beta \\
&= \sum_{\alpha, \beta, \gamma} p(\alpha) p(\beta) p(\gamma) ~ p(b|\alpha, \gamma) p(c|\beta, \gamma) \sigma^{A}_\alpha \otimes \sigma^{D}_\beta
\end{align}
where in the final line we extracted a \textit{local hidden variable} (LHV) model using the locality of $\rho^{B'C}$. Hence taking the central source $\rho^{B'C}$ as local, and the adjacent sources as separable will also always lead to an NLHS model, for any measurements.

Still taking $\rho^{AB}$ as separable as in Equation (\ref{eq:app_oppsep}), if instead now $\rho^{B'C}$ and $\rho^{C'D}$ are unsteerable towards $C$ and $D$ respectively, we have 
\begin{align}
\sigma_{b,c} &= \sum_\alpha p(\alpha)~\sigma_\alpha^A \otimes~ \text{Tr}_{CC'} \bigg ( \bigg [ M_c^{CC'} \otimes \mathbbm{1}_D \bigg ]  \text{Tr}_{BB'}\bigg ( \bigg [ M_b^{BB'} \otimes \mathbbm{1}_C  \bigg ]  \sigma^{B}_\alpha \otimes \rho^{B'C}  \bigg )  \otimes \rho^{C'D} \bigg ) \label{eq:app_unsteer1}\\
&= \sum_\alpha p(\alpha)~\sigma_\alpha^A \otimes~ \text{Tr}_{CC'} \bigg ( \bigg [ M_c^{CC'} \otimes \mathbbm{1}_D \bigg ]  \text{Tr}_{BB'}\bigg ( \bigg [ M_{b|\alpha}^{BB'} \otimes \mathbbm{1}_C  \bigg ]   \rho^{B'C}  \bigg )  \otimes \rho^{C'D} \bigg ) \\
&= \sum_\alpha p(\alpha)~\sigma_\alpha^A \otimes~ \text{Tr}_{CC'} \bigg ( \bigg [ M_c^{CC'} \otimes \mathbbm{1}_D \bigg ]   \bigg (\sum_\gamma p(\gamma) p(b|\alpha, \gamma) \sigma^C_\gamma  \bigg ) \otimes \rho^{C'D} \bigg ) \\
&= \sum_{\alpha,\gamma} p(\alpha)p(\gamma)~\sigma_\alpha^A \otimes~p(b|\alpha, \gamma)~ \text{Tr}_{C'} \bigg ( \bigg [ M_{c|\gamma}^{C'} \otimes \mathbbm{1}_D \bigg ]    \rho^{C'D} \bigg ) \\
&= \sum_{\alpha,\gamma,\beta} p(\alpha)p(\beta)p(\gamma)~p(b|\alpha, \gamma)p(c|\beta,\gamma)~  \sigma_\alpha^A \otimes \sigma_\beta^D \label{eq:app_unsteer2}
\end{align}
also leading to a NLHS model.

To summarise, ordering the sources as \{$\rho^{AB}$, $\rho^{B'C}$, $\rho^{C'D}$\} and denoting $\mathbf{SEP}$ as the set of separable states, $\mathbf{LOC}$ as the set of Bell-local states, and $\mathbf{UNSTEER}_\rightarrow$ as the set of unsteerable states (in an appropriate direction) we have that \{$\mathbf{SEP}$, ~$\mathbf{LOC}$, ~$\mathbf{SEP}$\}, \{${}_\leftarrow\mathbf{UNSTEER}$, ~$\mathbf{SEP}$, ~$\mathbf{UNSTEER}_\rightarrow$\}, \{$\mathbf{SEP}$, ~$\mathbf{UNSTEER}_\rightarrow$, ~$\mathbf{UNSTEER}_\rightarrow$\} and (by symmetry) \{${}_\leftarrow\mathbf{UNSTEER}$,  ~${}_\leftarrow\mathbf{UNSTEER}$, ~$\mathbf{SEP}$\} all admit NLHS models, for any measurements. This is captured in Figure \ref{fig:app_lhs_triangle}.

Recalling that there exist entangled yet unsteerable states, and steerable yet Bell-local states (that is $\mathbf{SEP}\subset \mathbf{UNS}_\rightarrow \subset \mathbf{LOC}$) demonstrates that these models are indeed non-trivial. Indeed network steering is truly a novel phenomena, and fully characterising the resources needed to demonstrate it is an open and fascinating new research question.

\subsection{General Line/Ring Networks}

We can generalise this to an arbitrary line network with trusted endpoints (Figure \ref{fig:app_line_scenario}), which again could be interpreted as a ring network with a single trusted party. 
\begin{figure}
    \centering
\includegraphics[]{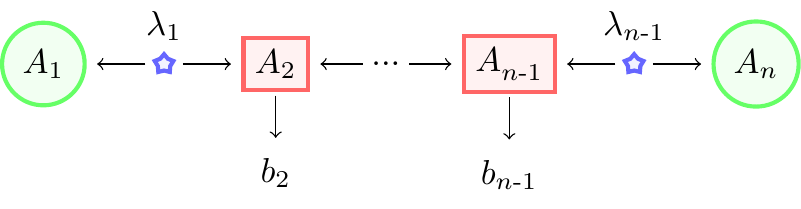}
    \caption{A general linear network with no inputs and trusted endpoints.}
    \label{fig:app_line_scenario}
\end{figure}

For $n$ parties, here an observed set of states would be described by
\begin{align}
\sigma_{b_2, \dots, b_{n-1}}^{A_1 A_n} =  \text{Tr}_{A_2A_2'\dots A_{n-1}A_{n-1}'}\bigg (\bigg [\mathbbm{1}^{A_1} \otimes M^{A_2A_2'}_{b_2} \otimes &\dots \otimes M^{A_{n-1}A_{n-1}'}_{b_{n-1}} \otimes \mathbbm{1}^{A_n} \bigg ] \nonumber \\ &\times\rho^{A_1 A_2} \otimes \rho^{A_2'A_3} \otimes \dots \otimes \rho^{A_{n-1}'A_n} \bigg ). \label{eq:app_linequantum}
\end{align}

The NLHS condition here generalises to
\begin{equation}
\sigma_{b_2, \dots, b_{n-1}}^{A_1 A_n}=\sum_{\lambda_1,\dots ,\lambda_{n-1}}
     p(\lambda_1)\dots p(\lambda_{n-1})
    \times p(b_2|\lambda_1,\lambda_2) \dots p(b_{n-1}|\lambda_{n-2}, \lambda_{n-1})
    \times \sigma^{A_1}_{\lambda_1}\otimes\sigma^{A_n}_{\lambda_{n-1}}. \label{eq:app_linelhs}
\end{equation}

We first remark that as stated in the main text, $\sum_{b_i} \sigma_{b_2, \dots, b_{n-1}}$ is a product state for any $b_i$, and the entanglement of a single $\sigma_{b_2, \dots, b_{n-1}}$ suffices to demonstrate network steering, being incompatible with Equation (\ref{eq:app_linelhs}).

From the previous calculations for the line with four parties (Equations (\ref{eq:app_triangle-first}) - (\ref{eq:app_unsteer1})), we see more generally how taking certain sources as separable can introduce natural sufficient conditions on the other sources to result in an NLHS model overall. For example if a single source is separable, then taking all other sources as unsteerable (in the direction away from this source) leads to an overall NLHS model for a line of any length -- this is a generalisation from the above Equations (\ref{eq:app_unsteer1}) to (\ref{eq:app_unsteer2}). Similarly, if a given source is unsteerable then upon receiving some input (for example from an adjacent separable source), the resulting LHS assemblage can serve as an input to the next party. This idea of ``percolation of inputs '' allows to write down a large class of NLHS models, for arbitrary linear networks. 

As a small example, we discuss the scenario in Figure \ref{fig:app_lhs_line} marked with ($\star$). The separable source second from the left provides an input to the adjacent sources, from which arises natural steering assemblages such as $\text{Tr}(M_{b|\lambda}^B \otimes \mathbbm{1}^{B'} \rho_{BB'})$. If these adjacent sources are steerable in the appropriate direction, we can extract an LHS model, whose corresponding state assemblages can act as an input to the next party. As the parties second and third from the right now receive effective inputs, the relevant condition on the second source from the right to admit a local model is of \textit{locality}. Therefore taking the sources as described would lead to an overall NLHS model for any measurements performed. These type of arguments would hold more generally for arbitrary linear network structures.

\begin{figure}[h]
    \centering
    \begin{subfigure}{.5\textwidth}
  \centering
\includegraphics[]{lhs_bilocal}
\caption{The simplest scenario.}
\label{fig:app_lhs_bilocal}
\end{subfigure}
\vspace{0.2cm}

\begin{subfigure}{.5\textwidth}
  \centering
  \includegraphics[]{lhs_triangle}

  \caption{The triangle scenario.}
  \label{fig:app_lhs_triangle}
\end{subfigure}
\vspace{0.2cm}

\begin{subfigure}{.5\textwidth}
  \centering
\includegraphics[]{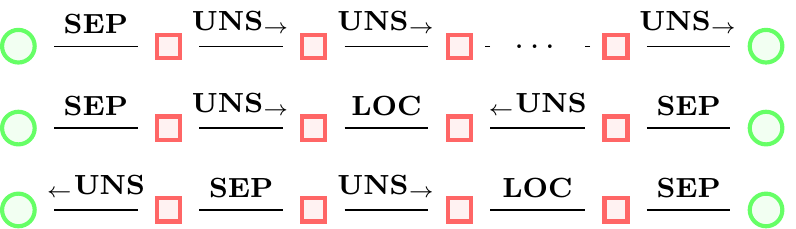} \hspace{10pt}\vspace{-10pt} ($\star$) \vspace{5pt}
  \caption{General scenarios and examples.}
  \label{fig:app_lhs_line}
\end{subfigure}

    \caption{Classes of NLHS models. Here green circles represent trusted parties, and red squares represent  untrusted parties, who perform a fixed measurement. By taking a source as separable ($\mathbf{SEP}$), this can act as an input to the adjacent measurements, and by taking further sources as unsteerable in a certain direction ($\mathbf{UNS}_\rightarrow$ ) or local ($\mathbf{LOC}$) as appropriate, we can arrive at an overall NLHS model for any possible measurements at the untrusted nodes. The example indicated by ($\star$) is discussed in the text.}  
    \label{fig:app_lhsmodels}
\end{figure}

\section{Entanglement Swapping of Doubly-Erased Werner States}
\label{app:ew-dew}

Here we will elaborate on and detail more closely how Doubly-Erased Werner (DEW) states can demonstrate network steering, despite being two way-unsteerable. Recall that the Erasure channel is given by
\begin{equation}
    \Lambda_\eta (\rho) = \eta \rho + (1-\eta)\text{tr}(\rho)\ketbra{d}{d}.
\end{equation}
For example, this channel acting on a qubit state would result in a qutrit state, where now the original qubit state $\rho$ is viewed as being embedded in the $\{ \ket{0}, \ket{1} \}$ subspace, and loss of the system is represented by the $\ket{2}$ state.

For $\rho^{AB}$ a two-qubit state, a result from \cite{tischler2018conclusive} states that $\Lambda_\eta \otimes \mathbbm{1} \rho^{AB}$ is unsteerable from Alice to Bob (for arbitrary measurements) if 
\begin{equation}
    \max_\mathbf{x} \bigg [(1-3\eta)|\mathbf{a}.\mathbf{x}| + \frac{3 \eta}{2}(1+(\mathbf{a} \cdot \mathbf{x})^2) + \norm{T\mathbf{x}} \bigg ] \leq 1.
\end{equation}
where $\mathbf{a}$ is Alice's local Bloch vector, $T$ is the bipartite correlation matrix with entries $T = \text{Tr}(\rho ~ \sigma_i \otimes \sigma_j)$ for $\sigma_i$ the Pauli matrices, and the maximisation is over unit vectors $\mathbf{x}$ in $\mathbbm{R}^3$. For $\rho^{AB}=\rho_W(\omega)= \omega \ketbra{\psi^-}{\psi^-} + (1-\omega)\mathbbm{1}/4$ the Werner state, we have $a=0$ and $T=\text{diag}(-\omega,-\omega,-\omega)$ and this condition becomes
\begin{equation}
  \eta \leq \frac{2}{3} (1- \omega). \label{eq:app_wsunsteerablecondition}
\end{equation}
Now as $\mathbbm{1}^A\otimes \Omega^B [ \rho^{AB} ]$  is unsteerable from Alice to Bob for any channel $\Omega$ if $\rho^{AB}$ is unsteerable from Alice to Bob \cite{quintino2015inequivalence}, we have that the the Doubly-Erased Werner (DEW) state
\begin{equation}
    	\rho_\text{\scalebox{.7}[1.0]{\tiny DEW}}(\eta, \omega) := \Lambda_\eta \otimes \Lambda_\eta \bigg ( \omega \ketbra{\psi^-}{\psi^-} + (1-\omega)\frac{\mathbbm{1}}{4} \bigg )
\end{equation}
is unsteerable in both directions for $\eta \leq \frac{2}{3} (1- \omega)$.

We now detail the full calculation of entanglement swapping for Doubly-Erased Werner (DEW) states. Expanding out the DEW state gives
\begin{align}
       \Lambda_{\eta} \otimes \Lambda_{\eta} \rho_W (\omega) =&\eta^2 \rho_W(\omega)
    + \eta (1-\eta) \frac{\mathbbm{1}_2}{2} \otimes \ket{2}\bra{2} \nonumber\\
    &+ \eta(1-\eta) \ket{2}\bra{2} \otimes \frac{\mathbbm{1}_2}{2} 
    + (1-\eta)^2  \ket{2}\bra{2} \otimes \ket{2}\bra{2}.
\end{align}

 Now consider entanglement swapping with projector $\ket{\psi^-}\bra{\psi^-}$ (on the $\{\ket{0}$, $\ket{1}\}$ subspace) onto two DEW states. We can write this as

\begin{align}
    \text{Tr}_{BB'}\bigg ( \mathbbm{1}^A \otimes\ket{\psi^-}\bra{\psi^-}^{BB'} \otimes \mathbbm{1}^C   &\bigg [ \Lambda_\eta \otimes \Lambda_\eta \rho_W (\omega)^{AB}   \bigg ]  \otimes  \bigg [ \Lambda_\eta \otimes \Lambda_\eta \rho_W (\omega)^{B'C}   \bigg ] \bigg ) \\
    = \text{Tr}_{BB'}\bigg ( \mathbbm{1} \otimes\ket{\psi^-}\bra{\psi^-} \otimes \mathbbm{1} & \bigg [ \eta^2 \rho_W (\omega) 
    + \eta(1-\eta) \ket{2}\bra{2}\otimes \frac{\mathbbm{1}_2}{2}  \\
    &  + \eta(1-\eta)  \frac{\mathbbm{1}_2}{2}\otimes\ket{2}\bra{2}  
     + (1-\eta)^2 \ket{2}\bra{2}\otimes \ket{2}\bra{2}\bigg ]^{\otimes 2} \bigg ).
\end{align}

Note that any term with $\bra{\psi^-}$ acting on a  $\ket{2}$ subspace vanishes, so we can simplify this to
\begin{align}
     &\text{Tr}_{BB'}\bigg ( \mathbbm{1} \otimes\ket{\psi^-}\bra{\psi^-} \otimes \mathbbm{1}  \bigg [ \eta^4 \rho_W (\omega)\otimes \rho_W (\omega) + \eta^3(1-\eta) \rho_W (\omega) \otimes \frac{\mathbbm{1}_2}{2} \otimes \ket{2}\bra{2} \nonumber\\
    &\hspace{110pt}+ \eta^3(1-\eta) \ket{2}\bra{2}\otimes \frac{\mathbbm{1}_2}{2}
    \otimes \rho_W (\omega) + \eta^2(1-\eta)^2  \ket{2}\bra{2}\otimes\frac{\mathbbm{1}_2}{2}\otimes \frac{\mathbbm{1}_2}{2}\otimes\ket{2}\bra{2} \bigg ] \bigg ) \label{eq:app_werneres1} \\
    &=\frac{1}{4} \bigg ( \eta^4 \rho_W(\omega^2) + \eta^3 (1-\eta) \frac{\mathbbm{1}_2}{2} \otimes \ket{2}\bra{2}  \label{eq:app_werneres2}\\
    &\hspace{30pt}+ \eta^3 (1-\eta) \ket{2}\bra{2} \otimes \frac{\mathbbm{1}_2}{2} + \eta^2(1-\eta)^2  \ket{2}\bra{2} \otimes \ket{2}\bra{2} \bigg ) \\
    &= \frac{\eta^2}{4} \Lambda_\eta \otimes \Lambda_\eta \rho_W (\omega^2) \\
    &=\frac{\eta^2}{4} \rho_\text{\scalebox{.7}[1.0]{\tiny DEW}}(\eta, \omega^2).
\end{align}
In lines (\ref{eq:app_werneres1}) - (\ref{eq:app_werneres2}) we used the fact that entanglement swapping of two Werner states leads to another Werner state with the product of the visibilities.

Hence entanglement swapping of two DEW states leads to another DEW state with the product of the original Werner visibilities. As discussed in the main text, the DEW state $\rho_\text{\scalebox{.7}[1.0]{\tiny DEW}}(\eta, \omega^2)$  is entangled for $\omega > \frac{1}{3}$, and so by choosing appropriate parameters we can witness network steering in a line of arbitrary length by entanglement swapping these unsteerable-yet-entangled states. Therefore network steering is a fundamentally different phenomenon to conventional quantum steering.

\section{Further Observations}
\label{app:further-obvs}
Here we detail two basic observations relating to general network scenarios that are not discussed in the main text.
\subsection{Endpoint sources between untrusted nodes with no inputs can be taken to be separable.}

We claim that in the following scenario (Figure \ref{fig:app_sep_wlog}),
\begin{figure}[!h]
  \centering
  \scalebox{1}{
\includegraphics[]{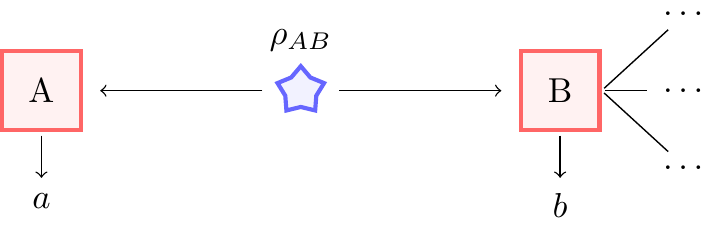}}
\caption{}
\label{fig:app_sep_wlog}
\end{figure}
we can take $\rho_{AB}$ to be separable without loss of generality. Suppose the overall state (or probabilities, if all other nodes are untrusted) is
\begin{align}
    \sigma_{a,b,\dots} &= \text{Tr}_{AB\dots} \bigg ( M_a^A \otimes M_b^{B\dots} \otimes \dots \bigg [ \rho_{AB} \otimes \dots \bigg ]\bigg ) \\
    &= \text{Tr}_{B\dots} \bigg (  M_b^{B\dots} \otimes \dots \bigg [ \text{Tr}_A \bigg (M_a \rho_{AB} \bigg ) \otimes \dots \bigg ]\bigg ).
\end{align}
Now set $\rho'_{AB}$ as 
\begin{equation}
    \rho'_{AB} = \sum_{a'} \ket{a'}\bra{a'} \otimes \text{Tr}_{A'}(M_{a'} \rho_{A'B})
\end{equation}
which is normalised. Then also set
\begin{equation}
    N_a^A = \ket{a}\bra{a}
\end{equation}
This gives
\begin{equation}
    \text{Tr}_{AB\dots} \bigg ( N_a^A \otimes M_b^{B\dots} \otimes \dots \bigg [ \rho'_{AB} \otimes \dots \bigg ]\bigg ) = \text{Tr}_{B\dots} \bigg (  M_b^{B\dots} \otimes \dots \bigg [ \text{Tr}_A \bigg (M_a \rho_{AB} \bigg ) \otimes \dots \bigg ]\bigg ) 
\end{equation}
as before, reproducing the same assemblage/probabilities using a separable state (and projective measurement $M_a$).

\subsection{NLHS models on the line can be reproduced with separable states and separable measurements.}

First recall that a measurement on subsystems $A$ and $B$ is said to be \textit{separable} if each effect $M_x^{AB}$ can be written as a sum of tensor products, that is for each $M_x^{AB}$ there exist $A_i$ and $B_i$ such that $M_x^{AB} = \sum_i A_i \otimes B_i$. If a measurement is not separable it is said to be \textit{joint}.

We prove the statement for the general line network, which has quantum model and NLHS as defined in Equations (\ref{eq:app_linequantum}) and (\ref{eq:app_linelhs}) respectively. We can rewrite the NLHS model as
\begin{align}
    \noindent&\sigma_{b_2, \dots, b_{n-1}}^{A_1 A_n}=\sum_{\lambda_2,\dots ,\lambda_{n-2}}
     p(\lambda_2)\dots p(\lambda_{n-2})
    \quad p(b_3|\lambda_2,\lambda_3) \dots p(b_{n-2}|\lambda_{n-3}, \lambda_{n-2})
     \quad \sigma^{A_1}_{b_2 | \lambda_2}\otimes\sigma^{A_n}_{b_{n-1}|\lambda_{n-1}},
\end{align}
by setting $\sigma^{A_1}_{b_2 | \lambda_2} = \sum_{\lambda_1} p(\lambda_1)p(b_2|\lambda_1, \lambda_2)~\sigma_{\lambda_1}^{A_1}$ and similarly for $\sigma^{A_n}_{b_{n-1}|\lambda_{n-1}}$. These standard LHS assemblages can be prepared with separable states and commuting measurements \cite{kogias2015hierarchy, moroder2016maps}. We can set $\rho^{A_1 A_2}$ and $\rho^{A_{n-1}'A_n}$ to be these separable states in question, and set all other sources according to the separable states
\begin{equation}
    \rho^{A_i A_{i+1}} =\sum_{\lambda_i} p(\lambda_i) \ket{\lambda_i \lambda_i}\bra{\lambda_i \lambda_i}.
\end{equation}
We then set the first measurement $M_{b_2}^{A_2 A_2'}$ to be the separable measurement $\sum_{\lambda_2} N_{b_2|\lambda_2}\otimes \ket{\lambda_1}\bra{\lambda_1}$, where $N_{b_2|\lambda_2}$ are the commuting ones mentioned above, and similarly for the last measurement $M_{b_{n-1}}^{A_{n-1} A_{n-1}'}$ . Then all other measurements can be defined via
 \begin{equation}
M_{b_i}^{A_i A_i'}=\sum_{\lambda_{i-1}, \lambda_i}p(b|\lambda_{i-1}, \lambda_i)\ket{\lambda_{i-1} \lambda_i}\bra{\lambda_{i-1} \lambda_i},
 \end{equation}
 which are separable. We can then see that inserting these expressions into (\ref{eq:app_linequantum}) would yield the desired NLHS assemblage in (\ref{eq:app_linelhs}).

\end{document}